\documentclass[a4paper, 12pt]{article}
\usepackage{amssymb,amsmath,latexsym,color,graphicx,stmaryrd}
\usepackage{hyperref}

\newtheorem{dref}{Definition}[section] 
\newtheorem{theo}[dref]{Theorem}
\newtheorem{prop}[dref]{Proposition}
\newtheorem{remark}[dref]{Remark}
\newtheorem{ex}[dref]{Example}

\newtheorem{assumption}[dref]{Assumptions}

\def\Ai{\mathop{\rm Ai}\nolimits}

\def\Tr{\mathop{\rm Tr}\nolimits}
\def\WF{\mathop{\rm WF}\nolimits}

\def\e {\mu}

\def\indt{{t}}
\def\indE{{E}}

\def\E{E}

\def\nE{\mathrm{E}'}
\def\np{\mathrm{p}'}
\def\nt{\mathrm{t}'}
\def\nx{\mathrm{x}'}

\def\pd{\partial}

\def\R{\mathbb{R}}
\def\S{\mathbb{S}}
\def\C{\mathcal{C}}
\def\Alg{\mathcal{E}}

\def\O{{O}}

\def\pxE{\pi_{E}}
\def\pxt{\pi_{t}}

\def\ve{\varepsilon}

\def\supp{\mathop{\mathrm{supp}}\nolimits}

\title{Asymptotic solutions: glancing trajectories, Lagrangian singularities, and Bessel cylinders}

\author{Ilya BOGAEVSKII$^{1,2}$ \& Michel ROULEUX$^3$}
\date{}
\begin{document}

\maketitle

\vskip-12pt
$^1$Guangdong Technion --- Israel Institute of Technology, China

$^2$partially supported by Isaac Newton Institute for Mathematical Sciences and London Mathematical Society, UK

ibogaevsk@gmail.com

\vskip12pt
{$^3$Universit\'e de Toulon, Aix Marseille Univ, CPT, CNRS, France

michel.rouleux@univ-tln.fr}

\vskip24pt

\begin{abstract}
  We find a normal form for the simplest Lagrangian singularity appearing as the projection onto the space-energy of
  a solution to the Hamilton--Jacobi equation. Using this normal form we approximate asymptotic solutions of an equation
  $(\widehat{H}-E_0) \, u_h = f_h$ with a semiclassical distribution $f_h$ microlocalized on a Lagrangian manifold $\Lambda_0$
  when $E_0$ is a critical value of the restriction $H|_{\Lambda_0}$.
\end{abstract}

\medskip


\bigskip

\noindent
{\small \textbf{2020 MSC:} 35C, 41, 44A, 33



\section*{Introduction}

We investigate asymptotic solutions to the equation
\begin{equation}
\label{beqf}
\bigl( \widehat{H} - \E \bigr) u_h (x,\E) = f_h(x), 
\quad x \in \R^n, \quad \E \in \R
\end{equation}
where $\widehat{H}$ is a (pseudo)differential operator with smooth\footnote[4]{Here and further ``smooth'' means $C^\infty$.}   coefficients depending on $x$,
the energy $E$ is considered as a parameter of the problem, and the right hand side $f_h$ is microlocalized on a smooth
Lagrangian submanifold $\Lambda_0$, i.\,e.
$
f_h = \mathcal{K}_{\Lambda_0} (a),
$
$\mathcal{K}_{\Lambda_0}$ is the Maslov canonical operator on $\Lambda_0$, and the amplitude $a: \Lambda_0 \to \R$ is a smooth function with a compact support.
In the simplest case:
$$
f_h (x) = a(x) \exp{\frac{i S_0(x)}{h}}, \quad \Lambda_0 = \left\{ p = \pd_x S_0 \right\}
$$
where $S_0$ is a smooth function of $x$. We consider \textit{asymptotic} solutions ---
it means that the equation \eqref{beqf} holds modulo $O(h^N)$ as $h \to +0$ for any natural $N$, and they are not {\it a priori} related with
exact solutions. 

A survey of the problem is given in \cite{AnDoNaRo3}.
The case when $E$ is not a critical value of the restriction $H|_{\Lambda_0}$ is investigated there (Theorem 4).
In the present paper we study the simplest situation when $E$ is a nondegenerate local extreme value of $H|_{\Lambda_0}$
using another general formula for asymptotic solutions \cite[Theorem 2]{AnDoNaRo3}.

To this end, we use Lagrangian singularity theory finding a normal form of the projection of the solution to the Hamilton--Jacobi equation to the space-energy (Theorem \ref{thm1}). Further, we apply this result from Lagrangian singularity theory to obtain an approximation of an asymptotic solution to the equation \eqref{beqf}
(Theorem \ref{thm2}). We use the results and technique developed in \cite{AnDoNaRo3, DoNa1, NaTo}.


\bigskip

This paper is organized as follows:

\smallskip

In Section 1 we consider the case
where $\E$ is an local extreme value of $H|_{\Lambda_0}$ at a nondegenerate critical point. We consider the Lagrangian submanifold  $\Lambda$
in the cotangent bundle to the space-time, consisting of the solutions to the Hamilton--Jacobi equation starting from the initial submanifold
$\Lambda_0$. We find a normal form of its projection to the space-energy $(x,E)$ up to fibred canonical transformations preserving the critical energy level.
This enters the general classification of \cite{ZaMy}. Our main result is Theorem \ref{thm1}.

In Section 2, using Theorem \ref{thm1}, we approximate the asymptotic solution to the equation (\ref{beqf}) near the glancing trajectory at a
critical energy for $H|_{\Lambda_0}$. Our approximation is expressed in terms of Airy function and its
derivative. This is Theorem \ref{thm2}.

Section 3 is devoted to the two dimensional case, where $\Lambda_0$ is Bessel cylinder, the wave-front set of a radially symmetric solution
to Helmholtz equation $(-h^2\Delta-1)u_h=0$. We consider Hamiltonians positively homogeneous of degree 1, such as the conformal metric
in the cotangent space $H(x,p)=|p|/\rho(x)$. We do not
investigate closely the structure of the asymptotic solution, but content ourselves to describe the phase function parameterizing $\Lambda$
together with its density.

In Section 4 we assume we are given two glancing hypersurfaces $(F,G)$ in $T^*M_0$ at some point $z$. By a Theorem of Melrose, see \cite{Ho}, Vol.\,IV,
we know they are symplectically equivalent to $F=\{x_1=0\}$ and $G=\{g(x,p)=p_1^2-p_2-x_1\}$.
Then we look for
the Lagrangian manifolds $\Lambda$ transverse to $F$ at $z$ and such that
$v_g\in T_z\Lambda$. We obtain a partial classification
of these singularities at the level of second order Taylor expansions, recovering a result from \cite{ZaMy}.
This could be possibly applied to diffraction theory by an obstacle in $\R ^n$.

\section{Glancing trajectories and Lagrangian singularities}

\label{sect1}

\subsection{Glancing trajectories}

Let us consider the \textit{space-time} $\R^{n+1}_\indt$ and its cotangent bundle $T^\ast \R^{n+1}_\indt$ with the following coordinates
$$
(x,t) = (x_1, \dots, x_n, t) \in \R^{n+1}_\indt, \quad (p,E) = (p_1, \dots, p_n, E) \in T^\ast_{(x,t)} \R^{n+1}_\indt.
$$
These coordinates physically means the following: $x$ are space coordinates, $p$ are momenta, $t$ is a time, $E$ is an energy. The canonical $2$-form is
$$
d p \wedge d x - d E \wedge d t = d p \wedge d x + d t \wedge d E,
$$
$- \E$ is dual to $t$, $t$ is dual to $\E$. A smooth Hamiltonian
$$
H : T^\ast \R^n \to \R, \quad x \in \R^n, \quad p \in T^\ast_x \R^n
$$
(which is considered as a function on $T^\ast \R^{n+1}_\indt$) defines the Hamiltonian direction field
$$
\dot{x} = \partial_{p} H(p,x), \quad \dot{p}= - \partial_{x} H(p,x), \quad \dot{\E} = 0.
$$
Let
$$
\Lambda_0 \subset T^\ast \R^n,
$$
be a smooth Lagrangian submanifold, which is called \textit{initial}.
The union $\Lambda$ of all integral curves of the Hamiltonian direction field which pass through the points
$$
z=(p, x, t, E) \in T^\ast \R^{n+1}_\indt \,:\, (p,x) \in \Lambda_0, \; t=0, \; \E = H(p,x)
$$
is a smooth Lagrangian submanifold.

A point $z_0 = (p_0,x_0) \in \Lambda_0$ is called \textit{glancing} if it is critical for the restriction $H|_{\Lambda_0}$
(when non degenerate, such a glancing point is also called a {\it kiss} in the terminology of \cite{ElGr}.)
Equivalently this means
that the Hamiltonian vector field $v_{H_0}$ is tangent to $\Lambda_0$ at $z_0$.

An integral curve of the Hamiltonian direction field passing through the point
$$
(p_0, x_0, 0, E_0) \in T^\ast \R^{n+1}_\indt, \quad \E_0 = H(p_0,x_0)
$$
is called a \textit{glancing} trajectory. All points of a glancing trajectory are critical for the restriction
$$
\E|_\Lambda = H|_\Lambda
$$
with the same critical value.

\subsection{Lagrangian singularities}

Let us consider the \textit{energy-space} $\R^{n+1}_\indE \ni (x,\E)$ and the two natural projections
\begin{gather*}
\pxE: T^\ast \R^{n+1}_\indt \to \R^{n+1}_\indE, \quad z=(p,x,t,\E) \mapsto (x,\E),
\\
\pxt: T^\ast \R^{n+1}_\indt \to \R^{n+1}_\indt , \quad z=(p,x,t,\E) \mapsto (x,t).
\end{gather*}
The fibres of the both projections are Lagrangian and so the latter define the following two Lagrangian mappings
$$
\Lambda \hookrightarrow T^\ast \R^{n+1}_\indt \stackrel{\pxE}{\longrightarrow} \R^{n+1}_\indE,
\quad
\Lambda \hookrightarrow T^\ast \R^{n+1}_\indt \stackrel{\pxt}{\longrightarrow} \R^{n+1}_\indt.
$$
If the second Lagrangian mapping is degenerate at some point of $\Lambda$, i.\,e.
$\mathop{\mathrm{rank}}{d \left[ \pxt |_\Lambda \right] (z)}<n+1$, then the first one is degenerate as well because $t$ is a coordinate on $\Lambda$.

If the first Lagrangian mapping is degenerate at some point outside of the glancing trajectories,
i.\,e. $\mathop{\mathrm{rank}}{d \left[ \pxE |_\Lambda \right] (z)}<n+1$, then the first one is degenerate as well because $E$
is a local coordinate on $\Lambda$ in this case. The first Lagrangian mapping is degenerate at all points of the glancing trajectories.
We investigate its simplest singularity when the second Lagrangian mapping is not degenerate, i.\,e. away from focal points.

Our first main result is Theorem \ref{thm1} giving a normal form of this simplest Lagrangian singularity up to diffeomorphisms preserving the energy level.

\begin{theo}
\label{thm1}
Let the following conditions hold:
\begin{enumerate}
\item
\label{dn1}
the restriction $H|_{\Lambda_0}$ has a local non-degenerate extremum (maximum or minimum) $E_0$ at a point $z_0 = (p_0,x_0)$;

\item
\label{dn2}
$d H (z_0) \neq 0$;

\item
$z_\star = (p_\star, x_\star, t_\star, E_0) \in \Lambda$ is a point of the glancing trajectory $\C$ starting at the point $z_0$;

\item
\label{dn4}
the restriction
$$
\pxt|_\Lambda : \Lambda \to \R^{n+1}_\indt
$$
is a local diffeomorphism of a neighborhood of the point $z_\star$ onto a neighborhood of the point $(x_\star, t_\star)$, i.\,e. $z_*$ is not a focal point.
\end{enumerate}
Then there exists a local canonical diffeomorphism $\gamma$ of $T^\ast \R^{n+1}_\indt$ defined in a neighborhood of the point $z_\star$
and satisfying the following conditions:
\begin{itemize}
\item
$\gamma: z_\star \mapsto 0$;
\item
$\gamma$ sends the fibres of the projection $\pxE$ to each other;
\item
$\gamma$ sends the hypersurface $E=E_0$ to $E=0$;
\item
the Lagrangian submanifold $\Lambda' = \gamma (\Lambda)$ is given by the following equations
$$
S'(x,t) = t^3/3 + (x_2^2 + \dots + x_{n}^2) t,
\quad
\E = - \pd_{t} S',
\quad
p = \pd_x S'
$$
in a neighborhood of the point $\gamma(z_\star) = 0$.
\end{itemize}
\end{theo}

\begin{remark}
\label{rem1}
Any local canonical diffeomorphism satisfying the first three conditions of Theorem \ref{thm1} is written as
$$
\gamma: (p,x,t,\E) \mapsto \bigl( \np(p,x,t,\E),  \, \nx(x,\E), \, \nt(p,x,t,\E), \, \nE(x,\E) \bigr),
$$
where
$$
d \np \wedge d \nx + d \nt \wedge d \nE = d p \wedge d x + d t \wedge d \E,
$$
$$
\np(z_\star)=0,  \quad \nt(z_\star)=0, \quad \nx(x_\star,\E_0) = 0, \quad \nE(x,\E_0) = 0.
$$
It is well known (see e.\,g. \cite[18.5]{AVG}, \cite[(3.1)]{DoNa1})
that any such a diffeomorphism is affine in the fibres of $\pxE$ and uniquely defined by some smooth function $W(x,\E)$ and the changes
$\nx(x,\E)$, $\nE(x,\E)$:
\begin{equation}
\label{lift}
(\np,\nt) = (p,t)
\left(
\begin{array}{cc}
\pd_x \nx & \pd_\E \nx
\\
\pd_x \nE & \pd_\E \nE
\end{array}
\right)^{-1}
+ \bigl( \pd_x W(\nx, \nE),\pd_\E W(\nx, \nE) \bigr)
\end{equation}
\end{remark}

\begin{ex}
  
\noindent 1) $\Lambda_0$ is the vertical plane. Consider $\Lambda_0=\R ^n_p$ and $H(x,p)=p^2-x_1$ (Stark effect),
      $H|_{\Lambda_0}=p^2$ has a non degenerate maximum for $E=0$. The trajectories
  issued from $x=0$ are given by $x_1=t^2+2\eta_1t, x'=2\eta't, p_1=t+\eta_1, p'=\eta'$, where $\eta=(\eta_1,\eta')\in \R ^n$. We have $H|_\Lambda=\eta^2$.
  On the glancing trajectory
  $x_1=t^2, x_2=0, p_1=t, p'=0$, and $H=0$.  See \cite{AnDoNaRo3}, Sect.4 for the non-glancing case $\eta_1>0,\eta'=0$.

\noindent 2) $\Lambda_0$ is ``horizontal''. Let
\begin{gather*}
H = -p_1, \quad S_0(x) = x_1^3/3 + x_1(x_2^2 + \dots + x_n^2), \quad z_0=0,
\\
\Lambda_0 = \left\{ p = \partial_x S_0 \right\} = \left\{ p_1 = x_1^2 + \dots + x_n^2, \; p_2 = 2 x_1 x_2, \; \dots, \; p_n = 2 x_1 x_n  \right\}.
\end{gather*}
Then
\begin{gather*}
S(x,t) = (x_1+t)^3/3 + (x_1+t)(x_2^2 + \dots + x_n^2),
\\
\Lambda= \left\{ p_x = \partial_x S, \; \E = - \partial_t S  \right\}.
\end{gather*}
The trajectory starting at $z_0=0$ is  $p=0$, $x_1=-t$, $x_2=\dots=x_n=0$, $\E=0$; let $z_\star$ be its point for $t=t_\star$. Then the canonical transformation
$$
(p,x,t,E) \mapsto (p_1+\E, p_2, \dots, p_n, x_1+t_\star, x_2, \dots, x_n, t+x_1, \E)
$$
gives
$$
S'(x,t) = t^3/3 + (x_2^2 + \dots + x_{n}^2) t
$$
according to Theorem \ref{thm1}. Both examples are actually closely related by a canonical transformation. 
\end{ex}

\begin{remark} 
Our normal form of Theorem \ref{thm1} coincides with one of the normal forms
given in \cite[Theorem 2]{ZaMy} with respect to another equivalence group consisting of all canonical diffeomorphisms preserving the function $E$.
In the present paper we preserve the projection $\pxE$ together with one fixed level of $E$,
and get only one normal form of \cite[Theorem 2]{ZaMy} due to the condition that $E_0$ is
a local degenerate maximum or minimum of $H|_{\Lambda_0}$. Otherwise minus signs may appear in the expression $x_2^2 + \dots + x_{n}^2$.
\end{remark}

\begin{remark} 
\label{rfam}
In terms of generating families (phase functions) --- see e.\,g. \cite{AVG, Iv}, or Subsection \ref{genfam} of the present paper ---
Theorem \ref{thm1} is reformulated in the following way. The Lagrangian submanifold $\Lambda$ in a neighborhood of the point $z_\star$
is defined by a generating family
$$
\Phi(\theta, x, \E) = S(x,\theta) + \E \theta, \quad \Lambda = \left\{ \pd_\theta \Phi =0, \; p=\pd_x \Phi, \; t=\pd_\E \Phi\right\},
$$
and there exist a smooth reversible change of coordinates
$$
(\theta,x,\E) \mapsto \bigl( \theta'(\theta,x,\E)),  \, \nx(x,\E), \, \nE(x,\E) \bigr)
$$
and a smooth function $W(x, \E)$ such that
\begin{itemize}
\item
$\theta'(t_\star, x_\star, \E_0) = 0, \quad \nx(x_\star,\E_0) = 0, \quad \nE(x,\E_0) = 0$;
\item
$\Phi(\theta,x,\E) = \Phi'(\theta', \nx,\nE)$ in a neighborhood of the point $z_\star$ where
$$
\Phi'(\theta, x, \E) = S'(x,\theta) + \E \theta - W(x,\E), \;\; S'(x,\theta) = \theta^3/3 + (x_2^2 + \dots + x_{n}^2) \theta.
$$
\end{itemize}
\end{remark}

\subsection{Proof of Theorem \ref{thm1}}

Let us recall some basic facts from Singularity Theory, see e.g. \cite{AVG}, \cite{Dui1}, Sect.2.
  We start to guess at a convenient
  form of the generating function (or generating family)
  $\Phi(\theta, x, \E)$ ($\theta\in{\bf R}^k$ being some ``phase variables'') that we take diffeomorphically to a ``normal form''
  preserving the critical set in the sense of Remark \ref{rfam}. More precisely, let $\Omega_{x,E}\subset\R^{n+1}_\mathrm{E}$
  and $\Omega_\theta\subset\R ^k$) be open sets, and $\Phi:\Omega_\theta\times \Omega_{x,E}\to \R$ a smooth function. Similarly, let
  $\Phi':\Omega_\theta\times \Omega_{x,E}\to \R$. 
  The functions $\Phi(\theta,x,E)$ and $\Phi'(\theta',x',E')$ are said
  {\it equivalent unfoldings as functions} on $\Omega_\theta$, which we denote by $\Phi'\sim\Phi$, iff there exists
  a diffeomorphism $\chi:\Omega_\theta\times \Omega_{x,E}\to \Omega_\theta\times \Omega_{x,E}$,
  $(\theta,x,E)\mapsto (\theta'(\theta,x,E),x'(x,E),E'(x,E))$ and $\psi\in C^\infty(\Omega_{x,E})$ such that
  \begin{equation}\label{11}
    \Phi'\circ\chi=\Phi+\psi
    \end{equation}
  (we stress that $\psi$ should not depend on $\theta$). In our case we shall also require that $\chi$ preserves the energy set,
  namely $E'(x,E)=E$.
  
  A function $\Phi\in C^\infty(\Omega_\theta\times \Omega_{x,E})$ is called a {\it stable unfolding} if there exists a neighborhood $U$ of $\Phi$ in
  $C^\infty(\Omega_\theta\times \Omega_{x,E})$ (for the Whitney $C^\infty$ topology) such that $\Phi'\sim\Phi$ for every $\Phi'\in U$.

  The first step in the study of stability of the unfolding is its reduction to a tranversality condition in the space of jets (Taylor expansions).
  We then pass from the level of jets  to this of smooth functions through the {\it deformation method}
  solving homological equations. 

  Unfolding of Lagrangian singularities requires that all such diffeomophisms preserve the symplectic structure.
  Such normal forms have been extensively studied, see \cite{AVG}, \cite{Dui1}, Sect.2, \cite{ZaMy}. See also \cite{KaRo} for the special matrix case
  of a hyperbolic umbilic.\\

  Let us now proceed to the proof of Theorem 1.1.
  In the present case, the Lagrangian singularity is a {\it fold} for the projection $\pi_E|_\Lambda:\Lambda\to \R ^{n+1}_E$.
 Most of our generating families induce only changes of variables in $(x,E)$, see e.g. Remark \ref{rem1}. 

Any point of the trajectory $\C$ is a local (non-strict) extremum of the restriction
$$
\E|_\Lambda = H|_\Lambda
$$
with the same extremum value $E_0$  because $H(g^t z) = H(z)$ for all $z \in T^\ast \R^n$.
All these critical points are degenerate of rank $n$ according to the condition \ref{dn1} of Theorem \ref{thm1}.
The condition \ref{dn4} of Theorem \ref{thm1} implies that $(x,t)$ are local coordinates on $\Lambda$ in a neighborhood of the point $z_\star$.
Let us consider the function
\begin{equation*}
\label{restr}
R (x,t) = \E|_{\Lambda}
\end{equation*}
which attains its local maximum $E_0$ along the curve
\begin{equation*}
\label{cX}
\pxt(\C)= \left\{ (x,t) \in \R^{n+1}_\indt \; | \;  x = X(t) \right\} \ni (x_\star,t_\star).
\end{equation*}
Its second differential $d^2 R(x_\star,t_\star)$ is semi-definite and has rank $n$ due to the condition \ref{dn1} of Theorem \ref{thm1}.

Our next observation is that
$$
\dot{X} (t_\star) = \pd_p H(z_\star)  \neq 0.
$$
Indeed, the trajectory $\C$ of the Hamiltonian vector field does not contain singular points
because $z_0$ is not singular by the condition \ref{dn2} of Theorem \ref{thm1}. So $d H(z_\star) \neq 0$ and the kernel of
$d H (z_\star)$ is tangent to $\Lambda$ because ${z_\star}$ is critical for the restriction $H|_\Lambda$.
Hence, this kernel is transversal to the plane $t=t_\star$, $x= x_\star$ due to the condition \ref{dn4} of Theorem \ref{thm1}.
Therefore, $\pd_p H({z_\star}) \neq 0$.

Hence,
$$
\pd_t^2 R(x_\star,t_\star) \neq 0.
$$
Indeed, $d^2 R(x_\star,t_\star)$ is semi-definite of rank $n$, but the vector $( \dot{X} (t_\star), 1 ) \neq (0, 1 )$
belongs to its kernel because $R(X(t),t) \equiv E_0$  (and hence, the vector $(0,1)$ does not belong to the kernel).
Therefore, in the coordinates $(x,t)$ on $\Lambda$
$$
\pxE |_\Lambda: \Lambda \to \R^{n+1}_\indE, \quad (x,t) \mapsto (x,R(x,t)),
$$
and the restriction $\pxE |_\Lambda$ is a fold in a neighborhood of the point $z_\ast$.


Let $\pxE(\Lambda, z_\star)$ be the image of the germ of $\Lambda$ at the point $z_\star$.
It is the germ at $\pxE(z_\star) = (x_\star,E_0)$ of a smooth $(n+1)$-dimensional submanifold with boundary.
The hypersurface $\Alg = \{E=E_0\} \subset \R^{n+1}_\indE$ is tangent to $\pxE(\Lambda, z_\star)$ with rank $n-1$ along the curve germ $\pxE(\C,z_\star)$.
The latter is smooth because $\dot{X}(t_\star) \neq 0$.

Besides, the restriction $\pxE |_\Lambda$ has Lagrangian singularity $A_2$ at $z_\star$ because it is a fold --- see e.\,g. \cite[21.3]{AVG}. So there exists a canonical local diffeomorphism
$$
\gamma_1: (p,x,t,\E) \mapsto \bigl( \np_1(p,x,t,\E),  \, \nx_1(x,\E), \, \nt_1(p,x,t,\E), \, \nE_1(x,\E) \bigr), \quad z_\ast \mapsto 0
$$
reducing $\Lambda$ to its local normal form $A_2$
$$
\Lambda_1 = \gamma_1(\Lambda) = \{E=-t^2, p=0\}.
$$
Here $\nE_1(x_\star,\E_0) = 0$ but $\nE_1(x,\E_0) \neq 0$ and the proof is not finished yet.

In order to finish the proof let us observe that the diffeomorphism $\gamma_1$ sends $\pxE(\Lambda, z_\star)$ to $\pxE(\Lambda_1, 0) = (E \leqslant 0, 0)$,
and the hypersurface $\Alg$ to a hypersurface $\Alg_1$ which is tangent to $\pxE(\Lambda_1, 0)$ with rank $n-1$ along a smooth curve germ.
According to the Morse lemma with parameters \cite[11.1]{AVG} there exists a local diffeomorphism
$$
(x,E) \mapsto (\nx_2(x),E), \quad 0 \mapsto 0
$$
reducing $\Alg_1$ to its local normal form $\Alg_2 = \{ E = x_2^2 + \dots + x_n^2\}$. Therefore, the canonical lift $\gamma_2$
of this diffeomorphism according to the formula \eqref{lift} with $W \equiv 0$ preserves $\Lambda_1$ and sends $\Alg_1$ to $\Alg_2$. At last, the canonical lift $\gamma_3$ of the diffeomorphism
$$
(x,E) \mapsto (x,E - x_2^2 - \dots -x_n^2)
$$
according to the formula \eqref{lift} with $W \equiv 0$ sends $\Lambda_1$ to $\Lambda'$ and $\Alg_2$ to $E=0$.
Therefore, the composition $\gamma_3 \circ \gamma_2 \circ \gamma_1$ of local canonical diffeomorphisms satisfies the conditions of Theorem \ref{thm1}.

\section{Glancing trajectories and asymptotic solutions}

\label{sect2}

\subsection{Approximations of asymptotic solutions}

We start to outline our framework of $h$-$\Psi$DO's and Maslov canonical operator.

Let $M_0={\bf R}^n$, or possibly a smooth manifold. For $m,k\in{\bf Z}$, 
we recall the usual class, writing $\langle\theta\rangle=(1+\theta^2)^{1/2}$,
\begin{equation}\label{1.1}
  S^k_m(M_0\times {\bf R}^N)=\{a(x,\theta;h)\in C^\infty: |\partial_{x}^\alpha\partial_\theta^\beta a(x,\theta;h)|\leq
  C_{\alpha,\beta}h^k \langle \theta\rangle^{m-|\beta|}\}
  \end{equation}
with asymptotic expansion $a(x,\theta;h)\sim h^k\bigl(a_0(x,\theta)+ha_1(x,\theta)+\cdots\bigr)$, see e.g. \cite{Iv}.
We call $a$ an {\it amplitude}, and shall sometimes consider amplitudes compactly supported in $\theta$. 
The space of {\it symbols} $H(x,p;h)$, $z=(x,p)\in T^*M_0$, or possibly $z=(x,p)\in T^*M_0\setminus 0$, with $p$ as a ``$\theta$-variable'',
is defined analogously, and allow Hamiltonians 
$H=H(x,p)$ positively homogeneous of degree $m$ with respect to $p$. 
We often omit the notation $S^k_m(\cdots)$ and
refer simply to ``amplitudes'' or ``symbols'' whenever this is clear from the context. 

Let 
$H(x,p;h)\in S^0_m(T^*M_0)$ so that $H(x,p;h)$ has the semi-classical expansion
\begin{equation}\label{1.2}
H(x,p; h)\sim H_0(x,p)+\,hH_1(x,p)+\cdots,\ h\rightarrow 0.
\end{equation}
We call as usual $H_0$ the principal symbol, and $H_1$ the sub-principal symbol.
We also assume that $H+i$ is elliptic.

This allows to take suitable quantizations of $H$, e.g. Weyl quantization
\begin{equation}\label{1.3}
H^w(x,hD_x;h)u(x;h)=(2\pi h)^{-1}\,\int\int e^{\frac{i}{h}\,(x-y)\,\eta}\,H\big(\frac{x+y}{2},\eta;h\big)\,u(y)\,dy\,d\eta,
\end{equation}
For short we write $\widehat H=H^w(x,hD_x;h)$. 

We investigate asymptotic solutions to the equation
\begin{equation}
\label{beq}
\bigl( \widehat{H} - \E \bigr) u_h (x,\E) = f_h(x), 
\end{equation}
\begin{equation*}
\widehat{H} = H (\widehat{p},x), \quad  H \in C^\infty( T^\ast \R^n), \quad \widehat{p}= -i h \partial_x, \quad x \in \R^n, \quad \E \in \R
\end{equation*}
where $H(x,p;h)\in S^0_m(T^*M_0)$, the energy $E$ is considered as a parameter
of the problem, and the right hand side $f_h$ is microlocalized on a smooth Lagrangian submanifold $\Lambda_0$, i.\,e.
$$
f_h (x) = \left[\mathcal{K}_{\Lambda_0} (a) \right] (x), \quad a: \Lambda_0 \to \R
$$
$\mathcal{K}_{\Lambda_0}$ is the Maslov canonical operator on $\Lambda_0$, and the amplitude $a$ is a smooth function with compact support.

Locally, $f_h$ is a Lagrangian semi-classical distribution of the form 
\begin{equation}\label{1.4}
  f_h(x)=(2\pi h)^{-N/2} \int_{{\bf R}^N} e^{i\varphi(x,\theta)/h}a(x,\theta;h)\,d\theta
  \end{equation}
$\varphi$ is a non-degenerate phase function in the sense of H\"ormander
(or a generating family in the sense of Arnold) defining the Lagrangian manifold $\Lambda_0$, and 
$a\in S^k_0(M_0\times{\bf R}^N)$ for some $k$. We also assume here that $a$ is compactly supported in $\theta$. 
Here are some examples:

\begin{equation}
\label{1.5}
f_h (x) = e^{-\pi i n / 4} h^{-\frac{n}{2}} \tilde{a} \left(\frac{x}{h} \right), \quad \Lambda_0 = \left\{ x = 0 \right\}
\end{equation}
Here $\tilde{a}$ is the unitary Fourier transform of the amplitude $a(p)\in S^0_0({\bf R}^n)$, i.e.
\begin{equation*}
e^{\pi i n / 4}   f_h (x) = {\cal F}_h^*a(x)=(2\pi h)^{-\frac{n}{2}} \int_{\R ^n} e^{ip(x-x_0)/h} a(p)\,dp=h^{-\frac{n}{2}} {\cal F}_1^*a(\frac{x}{h})
\end{equation*}
where ${\cal F}_h$ denotes (unitary) Fourier transform, and ${\cal F}_h^*$, defined by  ${\cal F}_h^*u(p)={\cal F}_h^*u(-p)$
its adjoint.
In this case $\Lambda_0=T^*_{x_0} \R^n$ is the vertical plane $x=x_0$ in the cotangent bundle $T^\ast \R^n$.
For simplicity we shall often assume $x_0=0$.

We can replace $a(p;h)$ in (\ref{1.5}) by a more general amplitude $a(p,x;h)$,
but it is known (\cite{Ho} Lemma 18.2.1) that we can eliminate the $x$ variable, without changing 
$f_h(x)$ mod ${\cal O}(h^\infty)$ 
with a new amplitude $b(p;h)$. Amplitude $a(p;h)$ in (\ref{1.5})
will be generally assumed to be compactly supported in $p$. Symbols like $a(p)=1$ are too singular from our point of view.

We can also take WKB functions in Fourier representation of the form
$f_h(x)=\int^* e^{i(xp+S(p))/h}a(x,p;h)\,dp$, 
and $\Lambda_0=\{(-\partial_pS(p),p): p\in {\bf R}^n\}$, which we can also call ``vertical''. 
Here $\int^*$ is a shorthand for $(2\pi h)^{-\frac{n}{2}}\int_{\R ^n}$, omitting phase normalization factors as $e^{-\pi i n / 4}$.

In contrast WKB symbols in the position representation of the form  
\begin{equation}\label{1.7}
f_h (x) = a(x) \exp{\frac{i S_0(x)}{h}}, \quad \Lambda_0 = \left\{ p = \pd_x S \right\}.
\end{equation}
are microlocalized  on a manifold $\Lambda_0$ we may call ``horizontal''. 

The solution to the equation \eqref{beq} is given by
\begin{equation}\label{1.12}
  u_h(x,E)= \frac{i}{h} \int_0^\infty e^{-it(\widehat{H}-\E)/h}f_h(x)\,dt
\end{equation}
where the integral makes sense after replacing
$$
e^{-it(\widehat{H}-\E)/h} \mapsto e^{-it(\widehat{H}-\E)/h-t\ve/h}
$$
and taking the limit $\ve\to0^+$,
so to fulfill Sommerfeld radiation condition (whenever it applies, see e.\,g. \cite{AnDoNaRo3}).

In this work, we content ourselves with formal asymptotic solutions, which
are not {\it a priori} related with
exact solutions. See \cite{Bon,Ca,KlCa} for a more complete study involving Sommerfeld radiation condition, when $\widehat{H}$
is Helmholtz operator in two dimensions, and $f_h(x)$ is localized at some point, i.e. of the form (\ref{1.5}).
A survey of the problem in the framework of Maslov canonical operator
is given in \cite{AnDoNaRo3,Ro} when $\Lambda_0$ is the vertical plane, i.e. $f_h$ of the form (\ref{1.5}).

\bigskip

Following \cite{AnDoNaRo3}, we consider:
\begin{assumption}
\label{assum}
\phantom{.}
\begin{enumerate}
\item
\label{cnd1}
The amplitude $a$ has a compact support $\supp{a} \subset \Lambda_0$.
\item
\label{cnd2}
The phase flow $g^t$ of the Hamiltonian vector field $v_H$ generated by $H$
\begin{equation*}
g^t: T^\ast \R^n \to T^\ast \R^n, \quad \dot{x} = \partial_{p} H(p,x), \quad \dot{p}= - \partial_{x} H(p,x)
\end{equation*}
is defined for all $z \in \Lambda_0$, $t \ge 0$.
\item
\label{cnd3}
For any $R > 0$, there exists $t_R$ such that $\bigl| \pi \bigl( g^t(z) \bigr) \bigr| > R$ for all $t > t_R$ and $z \in \Lambda_0$, where
$$
\pi : T^\ast \R^n \to \R^n, \quad (p,x) \mapsto x
$$
is the natural projection.
In particular, all points of $\Lambda_0$ are not singular for the Hamiltonian vector field $v_H$: $d H (z_0) \neq 0$ for all $z_0 \in \Lambda_0$.
\end{enumerate}
\end{assumption}

According to \cite{AnDoNaRo3} Assumptions \ref{assum} imply that
\begin{equation}
\label{int}
u_h(x,\E) =
\frac{i}{h} \int\limits_{0}^{+\infty} [\mathcal{K}_\Lambda (A_h)](x,t) \, \exp{ \left\{ \frac{i \E t}{h} \right\} } \, d t
\end{equation}
is an \textit{asymptotic} solution to the equation \eqref{beq} where $\mathcal{K}_{\Lambda}$ is the Maslov canonical operator on $\Lambda$,
\begin{equation}
\label{transp}
A_h : \Lambda \to \R, \quad \Pi A_h = 0, \quad A_h|_{\Lambda_0} = a
\end{equation}
is the solution depending on the parameter $h$ to the transport equation $\Pi A_h = 0$, $\Pi$ is the transport operator on $\Lambda$ --- see \cite[5.3, A.5]{AnDoNaRo3} for details. It means that
\begin{equation*}
\bigl( \widehat{H} - \E \bigr) u_h (x,\E) = f_h(x) + \O(h^\infty), \quad h \to +0
\end{equation*}
where
$$
\O(h^\infty) = \O(h^N), \quad \forall \; N > 0, \quad h \to +0.
$$

Let us make a few comments on the integral representation (\ref{int}). For simplicity we focus here on the case where $\Lambda_0$
is the vertical plane, following \cite{AnDoNaRo3}.

Define $L_0=\{H_0(x_0,p)=E_0\}\subset\Lambda_0$ as the
set of initial data.
There are generally two contributions to $u_h(x;E)$ for $E$ near $E_0$. 
The set  $\supp a\times \{0\}$ contributes to its ``boundary part'',
while $L_0\times[0,+\infty[$ contributes to its ``wave part'' (or ``far-field''). Indeed, if $\supp a\times L_0=\emptyset$,
    $u_h(x;E)$ reduces to its boundary part, i.e.
    $$u_h(x;E)=ih^{(n/2)-1}\int_0^\infty[{\cal K}_0B](x,t;h)\,dt$$
    where $B$ is an amplitude supported on
    $\Lambda_0$, see \cite{AnDoNaRo3}, Theorem 1. This property reflects the well-known relation between frequency sets~:
    $\WF _hu_h\subset\WF _h (\widehat H u_h)\cup \{H_0(x,p)=E(p)\}$

An important case is when $\partial_p H_0(x,p)\neq0$ on $L_0$ at some point $z_0=(x_0,p_0)$, or equivalently when Hamilton vector field $v_{H_0}$
is transverse to $\Lambda_0$ at $z_0$. We say that we are in the case of {\it Lagrangian intersection}, or that
$z_0$ is a {\it non-glancing point}. Then $L_0$ is (locally) a smooth $n-1$ dimensional manifold,
parametrized by $\alpha\in{\bf R}^{n-1}$, and so is $L_E$ for $E-E_0$ small enough. 
Let $\Lambda_+(E)$ be the its flow out through $v_{H_0}$. 
It is known that the map $\iota_+:\Lambda_+(E)\to{\bf R}^{2n}_{x,p}$ is a Lagrangian embedding when $T^*\R ^n$ is endowed with its
natural symplectic structure with measure $dp\wedge\,dx$. 

The basic result in this case (see \cite{MelUh}, based on earlier statements in \cite{DuiHo})  is that $\widehat H$ (microlocally near $z_0$) can be taken to the normal form $hD_{x_n}$,
in such a way that $\Lambda_0$ (the vertical plane) is preserved. 
So let $\theta_T\in C^\infty({\bf R}_+)$ vanishing near $+\infty$ and $\theta_T(t)=1$ for $t\leq T$. Conditions 
imply that, in suitable coordinates, the solution of (\ref{beq}) in the form (\ref{1.12}) can be written as 
\begin{equation}\label{1.15}
  u_h(x)={i\over h}\int_0^\infty\theta_T(t)\,dt\int^* e^{i(x'\xi'+(x_n-t)\xi_n)/h}e^{itE/h} a(\xi;h)\,d\xi
  \end{equation}
Namely, provided $x_n\leq T/2$ (say) an integration by parts and a non-stationary phase argument in variables
$(t,\xi_n)$ show that
$$(hD_{x_n}-E)u_h(x)=f_h(x)+{\cal O}(h^\infty)$$
Given $x_n\neq0$, integrations by parts in $\xi_n$ show that an arbitrary small neighborhood of $t=x_n$ contributes to (\ref{1.15})
mod ${\cal O}(h^\infty)$. 
Note that the solution $u_h(x,E)$ of $(hD_{x_n}-E)u_h(x,E)=f_h(x)+{\cal O}(h^\infty)$ instead is simply obtained by sticking in the factor
$e^{iEt/h}$ into the integral (\ref{1.15}).
\begin{remark}
  If we write the initial amplitude in the form $a(x',x_n,\xi)$ (without eliminating $x$)
  we get instead
  \begin{equation*}
  u_h(x)={i\over h}\int_0^\infty\theta_T(t)\,dt\int^* e^{i(x'p'+(x_n-t)p_n)/h}e^{itE/h}a(x',x_n-t,p;h)\,dp
  \end{equation*}
 \end{remark} 
In \cite{Ro} Theorem 1.2, we make this representation more precise, in terms of Maslov {\it bi-canonical operator}.\\

Representation (\ref{1.15}), as well as 
Theorem 4 of \cite{AnDoNaRo3} are not applicable in our situation because the principal type condition from
\cite{AnDoNaRo3} fails at glancing points.

So our second main result is Theorem \ref{thm2} giving an approximation for the asymptotic solution \eqref{int}
via the Airy function and its derivative in the coordinates of Theorem \ref{thm1} like it is done in \cite{DoNa1, NaTo}
for other Lagrangian singularities.

\begin{remark}
  Formula (2.8) from \cite[Theorem 2]{AnDoNaRo3} contains a cutoff function $\chi$ needed for convergence of the integral (2.8).
  But our condition \ref{cnd3} from Assumptions \ref{assum} is stronger that condition III from \cite[2.1]{AnDoNaRo3}
  and implies that the integral (2.8) from \cite[Theorem 2]{AnDoNaRo3} is well defined even if $\chi \equiv 1$.

  Formula (2.8) from \cite[Theorem 2]{AnDoNaRo3} is proved there if $E=0$ and $\Lambda_0$ is vertical.
  But its generalization for the case $E \neq 0$ is obvious.
  The case of an arbitrary Lagrangian submanifold $\Lambda_0$ does not have any obstacles as well --- see \cite[2.6.\,Discussion,\,4]{AnDoNaRo3}.

  But despite the above arguments it is not proved that the formula \eqref{int} gives an asymptotic solution in the described generality.
  So we formulate our Theorem \ref{thm2} for the integral \eqref{int} but not for an asymptotic solution of the equation \eqref{beq}.
\end{remark}

\begin{remark}
  A model problem when the condition \ref{cnd3} from Assumptions \ref{assum} fails is investigated in \cite{BoDoTo}.
\end{remark}

\begin{theo}
\label{thm2}
Let
\begin{equation*}
u_h(x,\E) =
\frac{i}{h} \int\limits_{0}^{+\infty} [\mathcal{K}_\Lambda (A_h)](x,t) \, \exp{ \left\{ \frac{i \E t}{h} \right\} } \, d t
\end{equation*}
be the function defined by \eqref{int}, \eqref{transp}, and the conditions \ref{cnd1}--\ref{cnd3} from Assumptions \ref{assum} hold with the following ones:
\begin{enumerate}
\item
The restriction $H|_{\Lambda_0}$ attains its a local non-degenerate extremum $E_0$ at a point $z_0 = (p_0,x_0)$.


\item
$z_\star = (p_\star, x_\star, t_\star, E_0) \in \Lambda$, $t_\star > 0$ is a point of the glancing trajectory $\C$ starting at the point $z_0$.

\item
\label{Cnd3}
The restriction
$$
\pxt|_\Lambda : \Lambda \to \R^{n+1}_\indt
$$
is a local diffeomorphism of a neighborhood of the point $z_\star$ onto a neighborhood of the point $(x_\star, t_\star)$.

\item
\label{Cnd4}
From the half-space $t>0$ the only point $z_\star$ is sent to the point $(x_\star, \E_0)$ by the  restriction
$$
\pxE|_{\Lambda } : \Lambda \to \R^{n+1}_\indE.
$$

\item
\label{Cnd5}
$x_\star \notin \pi(\supp{a})$.
\end{enumerate}
Then
in some local smooth coordinates
$$
x=X(y), \quad y=(y_1,\dots, y_n)
$$
in a neighborhood of the point $x_\star$:
\begin{multline*}
u_h \bigl( X(y), E_0 \bigr) = {}
\\
\left[ h^{-\frac23} b_h(y) \Ai \left( h^{-\frac23} r^2 \right)  + h^{-\frac13} c_h(y) \Ai' \left( h^{-\frac23} r^2 \right) \right] \exp{\frac{i w(y)}{h}}  +
\O(h^\infty),
\\
r^2 = y_2^2 + \dots + y_n^2, \quad h \to +0
\end{multline*}
where $w$, $b_h$, $c_h$ are smooth functions, $\Ai$ is the Airy function, and $\Ai'$ is its derivative.
\end{theo}

\subsection{Proof of Theorem \ref{thm2}}

The condition \ref{cnd3} from Assumptions \ref{assum}  implies that
\begin{equation}
\label{asyK}
[\mathcal{K}_\Lambda (A_h)](x,t) = \O(h^\infty)
\end{equation}
for $t \ge t_\mathrm{2} > 0$ uniformly on compacts in the $x$-space. Hence, for bounded $x$
\begin{equation}
\label{asyu}
u_h(x,\E) =
\frac{i}{h} \int\limits_{0}^{+\infty} \Theta(t) \, [\mathcal{K}_\Lambda (A_h)](x,t) \, \exp{ \left\{ \frac{i \E t}{h} \right\} } \, d t + O(h^\infty)
\end{equation}
where $\Theta$ is a cutoff function with a compact support $\supp{\Theta}$.

Besides, the condition \ref{Cnd5} of Theorem \ref{thm2} implies that the estimate \eqref{asyK} holds for $0 < t \le t_\mathrm{1}$ uniformly in any neighborhood of the point $x_\star$ not intersecting with $\pi(\supp{a})$. The\-re\-fore, we can choose $\Theta$ such that $\supp{\Theta} \subset (0,+\infty)$ and the estimate \eqref{asyu} holds for all $x$ from some neighborhood of $x_\star$.

The condition \ref{Cnd4} of Theorem \ref{thm2} and the condition \ref{cnd3} from Assumptions \ref{assum} show that for any neighborhood $U_{z_\star} \subset T^\ast \R^{n+1}_\indt$ of the point $z_\star$ there exists a sufficiently small neighborhood of the point $(x_\star, \E_0)$ such that its pre-image in the half-space $t > 0$ under the restriction $\pxE|_{\Lambda }$ is contained in $U_{z_\star}$. The\-re\-fore, we can choose $\Theta$ such that its support is contained in any neighborhood of $t_\star$ and the estimate \eqref{asyu} holds for all $(x,E)$ from some neighborhood of $(x_\star, \E_0)$.

Therefore
\begin{equation*}
u_h(x,\E) =
\frac{i}{h} \int\limits_{-\infty}^{+\infty}  [\mathcal{K}_\Lambda ( \varrho A_h )](x,t) \, \exp{ \left\{ \frac{i \E t}{h} \right\} } \, d t + O(h^\infty)
\end{equation*}
for some cutoff function $\varrho: \Lambda \to \R$ which can be chosen so that its support is contained in any neighborhood of the point $z_\star$. It means that the integral can be considered as the value of the precanonical operator in a local chart.

Let
\begin{gather*}
\gamma^{-1}: (p,x,t,\E) \mapsto \bigl( \np'(p,x,t,\E),   \, \nx'(x,\E), \, \nt'(p,x,t,\E), \, \nE'(x,\E) \bigr),
\\
\nE'|_{E=0} = \E_0
\end{gather*}
be the inverse diffeomorphism to $\gamma$ from Theorem \ref{thm1}. Then recomputing the precanonical operator in the coordinates $y = \nx (x, E)$, $\ve = \nE(x, \E)$ defined by Theorem \ref{thm1} (see Remark \ref{rfam} as well) we get
\begin{multline*}
u_h \bigl( \nx'(y,\ve), \, \nE'(y,\ve) \bigr) = {}
\\
\frac{i e^{{i W(y, \ve)}/{h}}}{h}  \int\limits_{-\infty}^{+\infty}  B_h(y,t) \, \exp{ \left\{ \frac{i ( t^3/3 + r^2 t  + \ve t ) }{h} \right\} } \, d t + O(h^\infty)
\end{multline*}
where the supports of all functions $B_h$ are contained in a neighborhood of the point $y=0$, $t=0$. If
$$
X(y) =  \nx'(y,0), \quad w(y) = W(y, 0).
$$
then
\begin{equation*}
u_h \bigl( X(y),E_0 \bigr) =
\frac{i e^{{i w(y)}/{h}}}{h}  \int\limits_{-\infty}^{+\infty}  B_h(y,t) \, \exp{ \left\{ \frac{i ( t^3/3 + r^2 t  ) }{h} \right\} } \, d t + O(h^\infty).
\end{equation*}
Choosing a large number $C>0$ and a cutoff function $\chi: \R \to \R$ such that: $\chi(t) = 1$ if $|t| \le C$ and $\chi(t) = 0$ if $|t| \ge C+1$, we get:
\begin{equation*}
\int\limits_{-\infty}^{+\infty}  B_h(y,t) \, \exp{ \left\{ \frac{i ( t^3/3 + r^2 t  ) }{h} \right\} } \, d t
=
\int\limits_{-\infty}^{+\infty} \chi(t) B_h(y,t) \, \exp{ \left\{ \frac{i ( t^3/3 + r^2 t  ) }{h} \right\} } \, d t
\end{equation*}
In order to finish the proof it is sufficient to show that for any function $B^i_h$ there exist functions $B^{i+1}_h$, $b^i_h$, and $c^i_h$ satisfying the equation
\begin{multline*}
\int\limits_{-\infty}^{+\infty} \chi(t)  B^i_h(y,t) \, \exp{ \left\{ \frac{i ( t^3/3 + r^2 t  ) }{h} \right\} } \, d t = {}
\\
2 \pi h^\frac13 b^i_h(y) \Ai \left(h^{-\frac23} r^2 \right) + 2 \pi h^\frac23 c^i_h(y) \Ai' \left(h^{-\frac23} r^2 \right) + {}
\\
 h \int\limits_{-\infty}^{+\infty} \chi(t) B^{i+1}_h(y,t) \, \exp{ \left\{ \frac{i ( t^3/3 + r^2 t  ) }{h} \right\} } \, d t + \O(h^\infty).
\end{multline*}
Indeed, let
\begin{multline*}
B^i_h(y,t) = {}^+\!B^i_h(y,t^2) +  {}^-\!B^i_h(y,t^2) t,
\\
b^i_h(y) = {}^+\!B^i_h(y,-r^2), \quad c^i_h(y) = {}^-\!B^i_h(y,-r^2);
\end{multline*}
then
$$
B_h^i(y,t) = b_h^i (y) + c_h^i(y) t + (t^2+r^2) \Delta_h(t,y), \quad B_h^{i+1}(y,t) = i \pd_t \Delta_h(t,y)
$$
and the following equalities prove our proposition:
\begin{multline*}
\int\limits_{-\infty}^{+\infty}  \chi(t) \exp{ \left\{ \frac{i ( t^3/3 + r^2 t  ) }{h} \right\} } \, d t =
\int\limits_{-\infty}^{+\infty}  \exp{ \left\{ \frac{i ( t^3/3 + r^2 t  ) }{h} \right\} } \, d t + \O(h^\infty) = {}
\\
h^\frac13 \int\limits_{-\infty}^{+\infty}  \exp{ \left\{ {i ( \tau^3/3 + r^2 \tau/ h^\frac23  ) } \right\} } \, d \tau + \O(h^\infty) =
2 \pi h^\frac13 \Ai \left(h^{-\frac23} r^2 \right) + \O(h^\infty),
\end{multline*}
\begin{multline*}
\int\limits_{-\infty}^{+\infty}  \chi(t) \,  t \exp{ \left\{ \frac{i ( t^3/3 + r^2 t  ) }{h} \right\} } \, d t =
\int\limits_{-\infty}^{+\infty}  t \exp{ \left\{ \frac{i ( t^3/3 + r^2 t  ) }{h} \right\} } \, d t + \O(h^\infty) = {}
\\
h^\frac23 \int\limits_{-\infty}^{+\infty} \tau \exp{ \left\{ {i ( \tau^3/3 + r^2 \tau/ h^\frac23  ) } \right\} } \, d \tau +
\O(h^\infty) = 2 \pi h^\frac23 \Ai' \left(h^{-\frac23} r^2 \right) + \O(h^\infty),
\end{multline*}
\begin{multline*}
\int\limits_{-\infty}^{+\infty}  \chi(t)  (t^2+r^2) \Delta_h(t,y) \exp{ \left\{ \frac{i ( t^3/3 + r^2 t  ) }{h} \right\} } \, d t = {}
\\
- i h \int\limits_{-\infty}^{+\infty}  \chi(t) \Delta_h(t,y)   d \exp{ \left\{ \frac{i ( t^3/3 + r^2 t  ) }{h} \right\} } = {}
\\
i h \int\limits_{-\infty}^{+\infty} \exp{ \left\{ \frac{i ( t^3/3 + r^2 t  ) }{h} \right\} }   \, d \left[ \chi(t) \Delta_h(t,y) \right]  = {}
\\
i h \int\limits_{-\infty}^{+\infty} \chi(t) \pd_t  \Delta_h(t,y) \exp{ \left\{ \frac{i ( t^3/3 + r^2 t  ) }{h} \right\} }  \, d t + {}
\\
i h \int\limits_{-\infty}^{+\infty}\Delta_h(t,y)  \pd_t  \chi(t)  \exp{ \left\{ \frac{i ( t^3/3 + r^2 t  ) }{h} \right\} }   \, d t  = {}
\\
h \int\limits_{-\infty}^{+\infty} \chi(t) B_h^{i+1}(y,t) \exp{ \left\{ \frac{i ( t^3/3 + r^2 t  ) }{h} \right\} }  \, d t + \O(h^\infty).
 \end{multline*}

Starting with $B^0_h = B_h$ and using this proposition again and again we get asymptotic series:
$$
b_h(y) = {2 \pi} \left( b^0_h(y) + h b^1_h(y) + h^2 b^0_h(y) + \dots \right) + O(h^\infty),
$$
$$
c_h(y) = {2 \pi} \left( c^0_h(y) + h c^1_h(y) + h^2 c^0_h(y) + \dots \right) + O(h^\infty).
$$
Any smooth functions with these asymptotic series satisfy the conditions of Theorem \ref{thm2}.


\section{The Bessel cylinder}

\label{BCyl}

In this section we consider the special case when the initial Lagrangian submanifold $\Lambda_0$ is the Bessel cylinder.

\subsection{Helmholtz equation and Bessel cylinder}

We discuss the physically relevant case of a Hamiltonian positively homogeneous with respect to $p$.
Such an Hamiltonian is given $H_0(x,p)=\frac{|p|^m}{\rho(x)}$, with $\rho$ a smooth positive function on $\R ^n$, and $m>0$.
When $m=2$ this is a conformal metric on the cotangent bundle $T^*\R ^n$.
Hamiltonians $H_0(x,p)=\frac{|p|^m}{\rho(x)}$ and $\widetilde H_0(x,p)=|p|^m-\rho(x)$
have the same trajectories up to a reparameterization of time
on the energy surfaces $\Sigma_1=\{H_0(x,p)=1\}$ and
$\widetilde\Sigma_0=\{\widetilde H_0(x,p)=0\}$ respectively.
When $m=2$, $\widetilde H_0$ defines the Helmholtz operator $h^2\Delta-\rho(x)$, see \cite{Ke, Ku}.

In case of variable coefficients it is convenient to consider instead Hamiltonians homogeneous of degree 1 in $p$ e.\,g.
change $H_0(x,p)=\frac{|p|^2}{\rho(x)}$ and its symmetric quantization
$\widehat H_0=\frac{1}{\sqrt{\rho(x)}}(-h^2\Delta)\frac{1}{\sqrt{\rho(x)}}$ to
$$
\widehat H_1=(\rho(x))^{-1/4}\sqrt{-h^2\Delta}(\rho(x))^{-1/4}
$$
using Weyl Calculus for $h$-pseudodifferential operators, see e.\,g. \cite{Iv}.

Most intrinsic formulae are available in that case, in particular we can construct
global eikonal coordinates and thus avoid the microlocal reduction to $H= \widehat{p}_n$  as in \cite{MelUh}.

Taking into account glancing intersection amounts
to ``correct'' locally the formulae giving the phase function and the half-density in new local coordinates,
which are obtained using a normal form. We follow here \cite{BoRo, Ro},  and references therein.

Let us consider the \textit{Bessel cylinder} $\Lambda_0\subset T^* \R^n$
  \begin{equation}
  \label{BC}
  \Lambda_0=\{x=\varphi\omega(\psi), \, p=\omega(\psi) : \varphi\in\R\}
  \end{equation}
  where $\omega(\psi)$ is a vector
  on the unit sphere in $\R^n$.

Our main motivation is the study of Bessel beams, i.\,e.
a wave whose amplitude is described by a Bessel function of the first kind that we will describe in Subsection \ref{applB} in more detail.

For the Helmholtz equation with variable coefficients, we have the following:

  \begin{prop} Let $\Lambda_0$ be the $n$-dimensional Bessel cylinder (\ref{BC}) and $H\in C^\infty(T^*\R^n)$ be homogeneous
    of degree $m$ with respect to $p$. Then $z=(x,p)\in\Lambda_0$ is a glancing point at energy $\E$ if and only if
    \begin{equation}
    \label{gl2}
    \partial_p H(z) + \varphi\partial_x H(z) = m H\omega(\psi), \; \langle-\partial_x H(z), \omega(\psi) \rangle=0, \; H(z)=\E
    \end{equation}
  \end{prop}

\noindent {\it Proof}: We complete $\omega(\psi)$ in $\S^{n-1}$ (the unit sphere) into a (direct) orthonormal basis
$\omega^\perp(\psi)=\bigl(\omega_1(\psi),\cdots,\omega_{n-1}(\psi)\bigr)$ of $\R^n$, and denote by
$\omega^\perp(\psi)\delta\psi=\omega_1(\psi)\delta\psi_1+\cdots+\omega_{n-1}(\psi)\delta\psi_{n-1}$ a section of $T \S^{n-1}$,
$\delta\psi_j\in\R$.
The tangent space $T_z\Lambda_0$ has the parametric equations
\begin{equation*}
  \delta X=\omega(\psi)\delta\varphi+\varphi\omega^\perp(\psi)\delta\psi, \
 \delta P=\omega^\perp(\psi)\delta\psi, \ (\delta\varphi,\delta\psi)\in \R ^n
\end{equation*}
so $v_H\in T_z\Lambda_0$ if and only if there exist $(\delta\varphi,\delta\psi)$ such that
\begin{equation*}
  \partial_pH=\omega(\psi)\delta\varphi+\varphi\omega^\perp(\psi)\delta\psi-\partial_xH=\omega^\perp(\psi)\delta\psi
  \end{equation*}
Taking scalar products with $\omega(\psi),\omega^\perp(\psi)$, and using that
$(\omega(\psi),\omega^\perp(\psi))$ form a basis of $\R^n$, readily shows that relations
\begin{equation}
\label{gl4}
  \partial_p H+\varphi\partial_x H=\langle\partial_p H,P(\psi)\rangle\omega(\psi), \
  \langle-\partial_x H,\omega(\psi)\rangle=0
  \end{equation}
are necessary and sufficient
for $v_H\in T_z\Lambda_0$.

We set ${\cal H}(\varphi,\psi)=H|_{\Lambda_0}$.
Then
\begin{equation*}
  \nabla{\cal H}(\varphi,\psi)=\bigl(\langle\partial_xH,\omega^\perp(\psi)\rangle, \varphi\langle\partial_xH,\omega^\perp(\psi)\rangle+
  \langle\partial_pH,\omega^\perp(\psi)\rangle\bigr)
  \end{equation*}
so (\ref{gl4}) readily gives
$\nabla{\cal H}(\varphi,\psi)=0$.
Therefore $z=(x,p)\in\Lambda_0$ is a glancing point.

Now, if $H$ is positively homogeneous of degree $m$ with respect to $p$,  using Euler identity, we get
$\delta\varphi=\langle\partial_p H,P(\psi)\rangle=mH$,
and (\ref{gl2}) holds if and only if
for $v_H \in T_z\Lambda_0$ when $H=\E$. $\Box$ \\

In particular when $H(z)=\frac{|p|^m}{\rho(x)}$ is a conformal metric, with $\rho$ a smooth positive function on $\R^n$,
$z$ is a glancing point if and only if
\begin{equation*}
    \begin{aligned}
    &\hbox{either}: \ \varphi\neq0 \ \hbox{and} \ \nabla\rho=0, \cr
    &\hbox{or}: \ \varphi=0 \ \hbox{and} \
    \langle\nabla\rho(0),\omega(\psi)\rangle=0.
    \end{aligned}
\end{equation*}

\begin{ex}\label{Ex3}: Let $n=2$, $m=1$, $H(z)=\frac{|p|}{\rho(x)}$, with
$\rho(x)=\frac{1}{2}(1+(x-x_0)^2)$. If $x_0=\varphi\omega(\psi)\neq0$, we have
$\rho^4(x_0)\det \nabla^2 (H|_{\Lambda_0})=\varphi^2$, $\rho^2(x_0)\Tr \nabla^2 (H|_{\Lambda_0})=-(1+\varphi^2)$.
  Critical energy is given by $\E=H(z)=\frac{|p|}{\rho(x)}$, i.\,e. $\E_0=1/\rho(x_0)$.
  \end{ex}

\subsection{Eikonal coordinates and generating families in the case $m=1$}
\label{genfam}

Let $\iota:L\to T^*\R^d_x$ be a smooth immersed Lagrangian manifold,
 the 1-form $p\,dx$ is closed on $L$, and so locally $p\,d x = d S$.
 If $p\,d x \neq 0$ then $S$ can be chosen as a (local) coordinate on $L$.
 Following \cite{DoMaNaTu, DoNaSh} we say that $S$ is the {\it eikonal} or the {\it action}
 on $L$ which can be completed to a system of {\it eikonal coordinates} on $L$.

 Let $L=\Lambda_0 \subset T^*\R^n$ the Bessel cylinder (\ref{BC}). Since $p\,dx=d\varphi$ on
 $\Lambda_0$ we get that $(\varphi,\psi)$ are eikonal coordinates.

Let  $L=\Lambda\subset T^*\R^{n+1}_\mathrm{t}$ be the integral manifold in the extended phase space. If $m=1$ then according to Euler
identity
$$\dot{S} = p \dot{x} - H = p \, \partial_p H - H = 0$$
along the trajectories and $(\varphi, \psi, t)$ are eikonal coordinates on $\Lambda$.

Let
$$\Phi: \R^N \times \R^d_x \to \R, \quad \theta \in \R^N$$
be a smooth function such that the $N\times (d+N)$ matrix
$\bigl( \partial^2_{\theta \theta} \Phi, \partial^2_{\theta x} \Phi \bigr)$
 has rank $N$ on the critical set
\begin{equation}
  \label{3.3}
  C_\Phi=\left\{ (\theta, x)\in \R^d\times \R^N: \partial_\theta \Phi (\theta, x)=0 \right\}.
  \end{equation}
In other words,
$C_\Phi$ defined by the equations (\ref{3.3}) is a smooth submanifold in the sense of the implicit function theorem.
Then
\begin{equation}
\label{3.4}
  \iota_\Phi:C_\Phi \to T^\ast\R^d_x, \quad (\theta, x) \mapsto  \bigl(x, \partial_x \Phi (\theta, x)\bigr)
  \end{equation}
is a Lagrangian (i.\,e. $\iota_\Phi^\ast d p \wedge d x = 0$) immersion,
$L=\iota_\Phi(C_\Phi)$ is an immersed Lagrangian submanifold, and $\Phi$ is called
is called a \textit{generating family} or \textit{phase function} of $L$ --- see e.\,g. \cite{AVG, Iv}.

It is standard to show that the number $N$ of $\theta$-variables can be reduced to $N \leq d$, and the minimal possible $N$
is equal to the co-rank of the projection $L \to \R _x^d$.

If the system (\ref{3.3}) is degenerate in the sense of the implicit function theorem then
$L=\iota_\Phi(C_\Phi)$ is a singular (isotropic) submanifold. In particular its dimension can be less than $d$. See e.\,g.
\cite{CdV} for a general discussion.

Let  $d=n$ and $L=\Lambda_0 \subset T^*\R^n$ be Bessel cylinder (\ref{BC}). According to private communications
with S.\,Dobrokhotov and V.\,Nazaikinskii, we use among the $\theta$-parameters
a Lagrange multiplier $\lambda$ and get for $\Lambda_0$ the generating family
\begin{multline*}
\Phi_0(x, \theta) = \varphi + \lambda \langle \omega(\psi), x - \varphi \omega(\psi) \rangle =
(1-\lambda) \varphi + \lambda \langle \omega(\psi), x  \rangle,\\
\theta = (\lambda, \varphi, \psi) \in \R^{n+1}.
\end{multline*}
Let now
$$p=P(\varphi, \psi, t), \quad x=X(\varphi, \psi, t)$$
be the Hamiltonian trajectory with an initial condition on the Bessel cylinder (\ref{BC}):
$$P(\varphi, \psi, 0) = \omega(\psi), \quad X(\varphi, \psi, 0) = \varphi \omega(\psi)$$
Apply then (\ref{3.4}) to $d=n+1$ and $L=\Lambda$ being the flow out of $\Lambda_0$ by $v_H$ in $T^\ast\R^{n+1}_\mathrm{t}$.

\begin{prop}
\label{prop2}
Let $H(x,p)$ be positively
homogeneous of degree 1 with respect to $p$.
Then
\begin{equation*}
\Phi(\theta, x, t) = \varphi + \lambda \langle P(\varphi, \psi, t) , x - X(\varphi, \psi, t) \rangle, \quad
\theta = (\lambda, \varphi, \psi) \in \R^{n+1}
\end{equation*}
is a generating family for
$\Lambda \subset T^\ast \R^{n+1}_\mathrm{t}$
at the points satisfying the inequality $\det(P,P_\psi)\neq0$,
which holds at least for small $t$ (because $\det(P,\partial_\psi P)=1$ for $t=0$). In particular, $\Phi$ satisfies
the initial condition $\Phi(\theta, 0, x)=\Phi_0(\theta,x)$.
\end{prop}

Note that $\Phi$ is the 1-jet on $\Lambda$
of the solution to the Hamilton--Jacobi equation
$$
\partial_t S + H(x,\partial_x S)=0, \quad  S|_{t=0}=\langle x,\omega(\psi)\rangle,
$$
namely
$\partial_t\Phi+H(x,\partial_x\Phi) = 0$ after the substitution $\lambda = 1$, $x = X(\varphi, \psi, t)$.

\subsection{Invariant density using the eikonal coordinates}

Let $y=(y_1,\cdots,y_d)$ be some local coordinates on $C_\Phi$ extended locally to smooth functions
on $\R^d\times\R^N$. Then the non vanishing real function

\begin{equation}\label{3.8}
F[\Phi,dy] = \frac{dy\wedge d(\partial_{\theta} \Phi)}{d x \wedge d\theta} =
\frac{dy_1 \wedge \dots \wedge d y_d \wedge d(\partial_{\theta_1} \Phi)\wedge
  \cdots\wedge d (\partial_{\theta_N} \Phi)}{d x_1 \wedge \dots \wedge d x_d \wedge d\theta_1\wedge\cdots\wedge d\theta_N}
\end{equation}
is well-defined near $C_\Phi$ as the quotient of two volume forms.

The restriction of this function to $C_\Phi$ is important for computations of the Maslov canonical operator on a Lagrangian submanifold
$L = \iota(C_\Phi)$ via its generating family $\Phi$ --- see e.\,g. \cite{DoMaNaTu, DoNaSh}.

\begin{prop} 
  Let $H(x,p)$ be positively homogeneous of degree 1 with respect to $p$, $\Phi$ be the generating family
  for $\Lambda\subset \R^{n+1}_\mathrm{t}$ from Proposition \ref{prop2}, $\theta=(\lambda,\varphi,\psi)$, and $y=(\varphi,\psi,t)$.
Then
\begin{equation}\label{3.9}
F[\Phi,dy]|_{C_\Phi}=
\pm\det(P,\partial_\psi P).
\end{equation}
\end{prop}
Details can be found in \cite{Ro}, \cite{BoRo}.

\subsection{Application to Bessel functions}

\label{applB}

We apply our considerations to Helmoltz operator with constant coefficients, thus retrieving integral formulas for Bessel functions.
Assume for instance we are given a non-linear
  Helmholtz equation on $\R^2$ of the type
  $$-h^2\Delta u-u=F(\e u), \quad F(0)=0, \quad F'(0)=1$$
  where $F$ is a smooth function and $\e$ is a small parameter.
  We expand $u=u_0+\e u_1+\cdots$, and find at zero order in $\e$ the equation $(-h^2\Delta-1)u_0=0$.
  Its radially symmetric solution is given by $u_0=f_h(x)$ where
  $$f_h(x)=(2\pi/h)^{1/2}J_0\bigl({|x| / h}\bigr),$$
  and is microlocalized on Bessel cylinder $\Lambda_0\subset T^* \R^2$ defined by (\ref{BC}) for $n=2$.
  Here
  $$J_0\bigl({|x| / h}\bigr)=\frac{1}{2\pi}\int_{-\pi}^\pi e^{i\langle x,\omega(\psi)\rangle/h}\,d\psi=\frac{1}{2\pi}\int_{-\pi}^\pi e^{i|x|(\sin\psi)/h}\,d\psi$$
  is Bessel function of order 0.

  At first order in $\e $ we get the Helmholtz equation
  $$(-h^2\Delta-1) u_1 =f_h, \quad H=p^2$$
  where $u_1 = u(x,1)$ in our previous notation, and
  $$u_h(x,1) = -(|x|/2 h)J_1(|x|/h) =-2(|x|/2h)^2\int_0^1tJ_0(t|x|/h)\,dt$$
  (where $J_1(z)=(2i\pi)^{-1}\int_{-\pi}^\pi e^{iz\cos\psi}e^{i\psi}\,d\psi$ is Bessel function of order 1)
  is its radially symmetric solution, which is again microlocalized on Bessel cylinder. We observe that all points of $\Lambda_0$ turn out
  to be glancing
  because $\Lambda_0 \subset \Sigma_1$ and $H_{\Lambda_0} = 1$.

  In case of constant coefficients we need not replace $p^2$ by $|p|$, and can use the generating family of Proposition 3.3.
  Consider first the homegeneous equation $(-h^2\Delta-1)u_h(x)=0$ in $\R ^2$. This amounts to choose the constant amplitude $a=1$ on $\Lambda_0$.
  We have
  $$P(t,\varphi,\psi)=\omega(\psi), \quad X(\varphi,\psi,t)=2t\omega(\psi)+\varphi\omega(\psi)$$
Choosing some local coordinates $y$ on $C_\Phi$, it is easy to see that
the (inverse) density on $\Lambda$ (up to sign) is given by $F[\Phi,\,dy]=\det(P,P_\psi)=1$, and
\begin{equation}\label{3.10}
  [{\cal K}_\Lambda a](x,t)= \frac{1}{2\pi h}\int e^{i\Phi(t,x,\theta)/h}F[\Phi,\,dy] a(y)\,d\theta
  \end{equation}
So the solution of $(-h^2\Delta-E)u_h=f_h$ is given by
$$
u_h(x)=\frac{i}{h}\int_0^\infty \,dt [{\cal K}_\Lambda a](x,t)e^{iEt/h}.
$$
Let $\Phi^E_+=\Phi+Et$, we compute the critical point
\begin{equation}\label{3.11}
  \begin{aligned}
&{\partial\Phi^E_+\over \partial\lambda}=\langle \omega^\perp(\psi),x-X(\varphi,\psi,t)\rangle=0, \quad {\partial\Phi^E_+\over \partial\varphi}=1-\lambda=0\cr
    &{\partial\Phi^E_+\over \partial\psi}=\lambda\langle \omega^\perp(\psi),x-X(\varphi,\psi,t)\rangle=\lambda\langle \omega^\perp(\psi),x\rangle=0
  \end{aligned}
\end{equation}
while ${\partial\Phi^E_+\over \partial t}=-2\lambda+E=-1$
when $E=\lambda=1$. Repeated integration by parts in (\ref{3.10}) with respect to  $t$ gives
\begin{equation}\label{3.12}
  u_h(x)=\frac{1}{2\pi h}\int e^{i\Phi^E_+(0,x,\theta)/h}\bigl(a+{h\over i}{\partial a\over \partial t}+\cdots\bigr) \,d\lambda\,d\varphi \,d\psi
  \end{equation}
Applying stationary phase to (\ref{3.12}) with respect to $(\lambda, \varphi)$, using that the critical value of
$\Phi_0=\Phi(0,x,\theta)$ is $\langle \omega(\psi),x\rangle$
readily shows that when $a=1$, we get only the first term, namely
$$u_h(x)=\int_{-\pi}^\pi e^{i\langle \omega(\psi),x\rangle/h}\,d\psi+{\cal O}(h^\infty)=2\pi J_0(|x|/h)+{\cal O}(h^\infty)$$

Note that the same kind of argument works for equation $(-h^2\Delta-1)u_h(x)=J_0(|x|/h)$, which gives
$u_h(x)=-{|x|\over h}J_1(|x|/h)$, etc\dots. In fact all these functions are microlocally supported on Bessel cylinder.

\section{Glancing hypersurfaces and La\-g\-ran\-gi\-an intersections}

We discuss here a possible application of non transverse Lagrangian intersections to problems of diffraction by an obstacle.
Let $M$ be a smooth manifold, since we are working locally, we will assume $M=\R^n$.
Let $F,G$ be two smooth hypersurfaces of $T^*M$,  intersecting
transversally at $z$. Recall \cite[Definition 21.4.6]{Ho} that $F$ and $G$ are said to be {\it glancing} at $z$ if and only if
the Hamilton foliation of $F=\{f=0\}$ and
$G=\{g=0\}$ (locally near $z$) are simply tangent at $z$.

Stated otherwise, we have $f(z)=g(z)=\{f,g\}(z)=0$ (Poisson bracket),
but the second Poisson brackets $\{f,\{f,g\}\}(z), \{g,\{g,f\}\}(z)$ are non zero
(we recall that Poisson bracket $\{f,g\}$ of $f=f(x,p)$, $g=g(x,p)$ is defined as $\{f,g\}=v_f g=-v_g f$, where as before $v_f$
denotes Hamilton vector field.~)

By the theorem of equivalence of glancing hypersurfaces of Melrose \cite[Theorem 21.4.8]{Ho}
there are local symplectic coordinates $(x,\xi)$ vanishing at $z$ such that
$F,G$ are defined resp. by $x_1=0$ and $\xi_1^2-x_1-\xi_2=0$.
Then $g=\xi_1^2-x_1-\xi_2=0$ will be the ``normal form'' of $H$ in these coordinates. (We use the notation $g$ for the normal form of $H$, or $H-\E$).

We apply this theorem to $G$ being the energy surface $H=\E$ (i.\,e. $g=0$) and $F$ an auxiliary hypersurface intersecting $G$ transversally at
a glancing point $z$. We want to find some germs of Lagrangian manifolds $\Lambda$ such that $\Lambda$ is transverse to $F$ at $z$
but $(\Lambda,G)$ has glancing intersection at $z$.
This means that $T_z\Lambda\cap(T_zF)^\sigma=\{0\}$ and $\R v_H(z)=(T_zG)^\sigma\subset T_z\Lambda$, where superscript $\sigma$
denotes symplectic orthogonal.
Assume $n=2$ for simplicity.

Define $\Lambda$ locally near $z$ by $f_1=f_2=0$, with $\{f_i,f_j\}=0$. Consider the symmetric matrix
$$
A_z=A_z(\Lambda,G)={}
\\
\begin{pmatrix}\{f_2,\{f_2,g\}\}&-\{f_1,\{f_2,g\}\}\\ -\{f_2,\{f_1,g\}\}&\{f_1,\{f_1,g\}\}\end{pmatrix}(z)
$$
and the vector
$$
B_z=B_z(\Lambda,G)=
\begin{pmatrix}\{g,\{g,f_1\}\} \\ \{g,\{g,f_2\}\} \end{pmatrix}(z).
$$
Let $z$ be a {\it glancing point} for the pair $(\Lambda,G)$ i.\,e.
$$
g(z)=f_1(z)=f_2(z)=0, \\
 \{g,f_1\}(z)=\{g,f_2\}(z)=0.
$$
We distinguish the following 10 possibilities for the 2-jets of $(\Lambda,G)_z$ (not covering the entire classification of \cite{ZaMy}):
\begin{enumerate}
\item
$\det A_z> 0, \ B_z\neq0$;
\item
$\det A_z > 0, \ B_z = 0$;
\item
$\det A_z < 0, \ {}^tB_z A_z B_z \neq 0$;
\item
$\det A_z < 0, \ {}^tB_z A_z B_z = 0, \ B_z\neq 0$;
\item
$\det A_z < 0, \ B_z = 0$;
\item
$\det A_z = 0, \ {}^tB_z A_z B_z \neq 0$;
\item
$\det A_z = {}^tB_z A_z B_z = 0, \ A_z \neq 0, \ B_z \neq 0$;
\item
$\det A_z = \ B_z = 0, \ A_z \neq 0$;
\item
$A_z = 0, \ B_z \neq 0$;
\item
$A_z = B_z = 0$.
\end{enumerate}

We know that a general Lagrangian manifold can be parameterized in the mixed representation, so when $n=2$ by one of the following cases
\begin{equation*}
  \begin{aligned}
(I)  \; & \Lambda=\left\{p=\partial_x\phi \right\}, \\
(II) \; & \Lambda=\left\{x=-\partial_p\phi \right\}, \\
(III)\; & \Lambda=\left\{x_1=-\partial_{p_1}\phi, p_2=\partial_{x_2}\phi \right\}, \\
(IV) \; & \Lambda=\left\{p_1=\partial_{x_1}\phi, x_2=-\partial_{p_2}\phi \right\}.
  \end{aligned}
\end{equation*}
To determine $\Lambda$, we write that
the glancing intersection of $(\Lambda,G)$ at $z$ should take place at $z=(x_1,x_2,p_1,p_2)=0$, i.e. for $z$ small enough,
$v_g(z)\in T_z\Lambda$ implies $z=0$. Then we check that $\Lambda$ is transverse to $F$ at $z=0$.

We have
\begin{prop}
Assume $\Lambda$ as above to be parameterized by a quadratic phase $\phi=\phi_0$.
In Cases (I), (II), (III) $\phi_0$ are one-parameter families taking values:
\begin{equation}
  \begin{aligned}
&(I) \quad \phi_0(x)=\frac12(ax_1^2-2x_1x_2)\cr
&(II) \quad \phi_0(p)=\frac12(2p_1p_2+cp_2^2)\cr
    &(III) \quad \phi_0(x_2,p_1)=\frac12b(p_1+x_2)^2
  \end{aligned}
  \end{equation}
respectively, where $a,b,c\neq0$. The manifold $\Lambda$ is transverse to $F$ at $z=0$,
and the corresponding matrices $A_z(\Lambda,G),B_z(\Lambda,G)$ above are then given by:
\begin{equation}\label{3.5}
  \begin{aligned}
&(I) \quad  A_z=2\begin{pmatrix}1&a\\ a&a^2\end{pmatrix}, \quad B_z=2\begin{pmatrix}-a \\ 1\end{pmatrix}\cr
&(II) \quad A_z=2\begin{pmatrix}0&0\\ 0&1\end{pmatrix}, \quad B_z=2 \begin{pmatrix} 1 \\ 0 \end{pmatrix}\cr
&(III) \quad A_z=2\begin{pmatrix}0&0\\ 0&1\end{pmatrix}, \quad B_z=2 \begin{pmatrix} 1 \\ 0 \end{pmatrix}
  \end{aligned}
  \end{equation}
So with the notations of Definition 1.3, case (I) is of type (6) or (7) according
to $a^3+2a^2+1=0$ or not, while cases (II) and
and (III) are of type (6). \\

At last, Case (IV) does not occur.
\end{prop}

\noindent {\it Proof}:

\smallskip

\noindent $\bullet$ Case (I). The quadratic phase takes the form $\phi_0={1\over2}(ax_1^2+2bx_1x_2+cx_2^2)$.
That $v_g$ be tangent to $\Lambda$ at $z$ express as
\begin{equation}\label{3.6}
  (\delta x_1,\delta x_2;a\delta x_1+b\delta x_2,b\delta x_1+c\delta x_2)=(2\xi_1,-1;1,0)
  \end{equation}
By substitution we find $p_1={b+1\over2a}={c\over2b}$, so condition $p_1=0$ gives $b+1=c=0$. Substituting in $p_1=ax_1+bx_2$
gives $ax_1-x_2=0$. Then the condition $z\in G$, namely
$g=0$ gives $-x_1-p_2=0$, so altogether $x_1=-p_2={x_2\over a},p_1=0$, and
$z=(x_1,ax_1,0,-x_1)$. Finally the condition $z\in F$ gives $x_1=0$, hence $z\in F\cap G$ gives $z=0$.
So the phase $\phi_0$ defining $\Lambda$ is given by $\phi_0={1\over2}(ax_1^2-2x_1x_2)$.
It is also clear that $\Lambda$ is transverse to $F$ at $z=0$.
Let us compute the matrix elements of $A|_{z=0}, B|_{z=0}$. With $f_j=p_j-{\partial\phi\over\partial x_j}$, we have
\begin{equation*}
  \begin{aligned}
    \{f_1&,\{f_1,g\}\}=2a^2, \ \{f_2,\{f_2,g\}\}=2, \ \{f_1,\{f_2,g\}\}=-2a\cr
      &\{g,\{g,f_1\}\}=-2a, \ \{g,\{g,f_2\}\}=2
  \end{aligned}
  \end{equation*}
which gives (\ref{3.5})(I).

\smallskip

\noindent $\bullet$ Case (II). The phase is of the form $\phi_0={1\over2}(ap_1^2+2bp_1p_2+cp_2^2)$, Hamilton vector field
$v_g=(-a\delta p_1-b\delta p_2,-b\delta p_1-c\delta p_2;\delta p_1,\delta p_2)=(2p_1,-1;1,0)$, so necessary condition $p_1=0$ gives
$a=0,b=1$. The condition $z\in G$ gives $z=(-p_2,-cp_2;0,p_2)$, and the condition $z\in\Lambda$ gives $-x_2=0$,
so we get $z=0$. The phase defining $\Lambda$ is given by $\phi_0={1\over2}(2p_1p_2+cp_2^2)$.
Again $\Lambda$ is transverse to $F$ at $z=0$. With
$f_j=x_j+{\partial\phi\over\partial p_j}$, we have
\begin{equation*}
  \begin{aligned}
    \{f_1&,\{f_1,g\}\}=2, \ \{f_2,\{f_2,g\}\}=0, \ \{f_1,\{f_2,g\}\}=0\cr
    &\{g,\{g,f_1\}\}=2, \ \{g,\{g,f_2\}\}=0
  \end{aligned}
  \end{equation*}
which gives (\ref{3.5})(II).

\smallskip

\noindent $\bullet$ Case (III). The phase is of the form $\phi_0={1\over2}(ap_1^2+2bx_2p_1+cx_2^2)$,
and Hamilton vector field
$v_g=(-a\delta p_1-b\delta x_2,\delta x_2;\delta p_1,b\delta p_1+c\delta x_2)=(2p_1,-1;1,0)$, so necessary condition $p_1=0$ gives $a=b=c$.
The condition $z\in G$ gives $z=(x_1,-{x_1\over a};0,-x_1)$, and the condition $z\in\Lambda$ gives $x_1=0$,
so we get $z=0$ as expected. The phase defining $\Lambda$ is thus given by $\phi_0={a\over2}(\xi_1+x_2)^2$.
Again $\Lambda$ is transverse to $F$ at $z=0$. With
$f_1=x_1+{\partial\phi\over\partial p_1}$, $f_2=p_2-{\partial\phi\over\partial x_2}$ we have
\begin{equation*}
  \begin{aligned}
    \{f_1&,\{f_1,g\}\}=2, \ \{f_2,\{f_2,g\}\}=0, \ \{f_1,\{f_2,g\}\}=0\cr
    &\{g,\{g,f_1\}\}=2, \ \{g,\{g,f_2\}\}=0
  \end{aligned}
  \end{equation*}
which gives (\ref{3.5})(III).

\smallskip

\noindent $\bullet$ Case (IV). The quadratic phase takes the form $\phi={1\over2}(ax_1^2+2bx_1p_2+cp_2^2)$,
but Hamilton vector field
$v_g=(\delta x_1, -b\delta x_1-c\delta p_2;a\delta x_1+b\delta p_2,\delta p_2)=(2p_1,-1;1,0)$ cannot be tangent to $\Lambda$ near $z=0$.
$\Box$\\

We can continue Taylor expansion of the phase functions to higher in $(x,p)$. Then we look (still in the perturbative sense)
for a solution to the Hamil\-ton--Ja\-co\-bi equation
$$\left({\partial\Psi\over\partial x_1}\right)^2-x_1-{\partial\Psi\over\partial x_2}=0$$
as the critical value with respect to a time-dependent phase function $\Phi$ that we solve in each of the cases (I)--(III).

\end{document}